\documentclass[amsmath,amssymb,reprint,floatfix]{revtex4-2}

\usepackage[T1]{fontenc}
\usepackage[utf8]{inputenc}
\usepackage{microtype}
\usepackage{mathtools}
\usepackage{amsthm}
\usepackage{bm}
\usepackage{booktabs}
\usepackage{array}
\usepackage{enumitem}
\usepackage{hyperref}
\usepackage{xcolor}

\hypersetup{colorlinks=true,linkcolor=blue,citecolor=blue,urlcolor=blue}
\allowdisplaybreaks
\newtheorem{theorem}{Theorem}[section]

\newtheorem{corollary}[theorem]{Corollary}
\theoremstyle{definition}
\newtheorem{definition}[theorem]{Definition}

\newcommand{\dd}{\mathrm{d}}

\begin{document}

\title{Spectral Invariance for Black-Hole Differential Observables}
\author{Vedant Subhash}
\affiliation{Department of Mathematics, University at Buffalo, The State University of New York, Buffalo, New York 14260, USA}
\makeatletter\def\Dated@name{}\makeatother\date{}

\begin{abstract}
We study first-order differential quantities built from solutions of second-order black-hole wave equations. These quantities can satisfy new differential equations that contain extra singular points, even when the original physical solutions remain regular. We show that an isolated zero of the transformation determinant can produce such an apparent singularity, but it does not by itself create a new resonance or quasinormal-mode frequency. We then study the corresponding Green functions and show that the physical transformed response has the same spectral poles as the original problem, provided that the transformation does not become globally degenerate. For the Regge--Wheeler and Zerilli equations, this means that ordinary turning points can appear as singular points of the derivative equation without changing the nonzero-frequency quasinormal-mode spectrum. The zero-frequency case is different and must be studied separately.

\end{abstract}


\maketitle

\section{Introduction}
\label{sec:intro}

Black-hole perturbation theory often reduces a field equation to a
second-order ordinary differential equation.  For a Schwarzschild black
hole, the odd- and even-parity gravitational perturbations are described
by the Regge--Wheeler and Zerilli equations
\cite{ReggeWheeler1957,Zerilli1970a,Zerilli1970b}.
Gauge-invariant formulations of the same problem were developed later,
and the relations between the different master variables are now well
understood
\cite{Moncrief1974,Chandrasekhar1975,Chandrasekhar1983}.
For curvature perturbations one instead obtains the Bardeen--Press
equation in Schwarzschild spacetime and the Teukolsky equation in Kerr
spacetime
\cite{BardeenPress1973,Teukolsky1972,Teukolsky1973,
PressTeukolsky1973}.
These equations remain basic tools in black-hole perturbation theory.

The radial master function is not always the final quantity of physical
interest.  Metric reconstruction, curvature quantities, flux variables,
and short-range reformulations can contain both a radial function and
its derivative.  It is therefore natural to study a first-order
differential observable of the form
\begin{equation}
 z=T(\omega)y
   =A(x,\omega)y+B(x,\omega)y',
 \label{eq:introT}
\end{equation}
where the original function satisfies
\begin{equation}
 L(\omega)y
 =
 y''+p(x,\omega)y'+q(x,\omega)y=0.
 \label{eq:introL}
\end{equation}
The functions $A$ and $B$ may depend on the spectral parameter
$\omega$.  The map $T$ does not have to be a Darboux transformation,
and in this paper we will not assume that it intertwines two operators.

Differential transformations already play an important role in
black-hole perturbation theory.  Chandrasekhar found relations between
the Regge--Wheeler, Zerilli, and Bardeen--Press equations
\cite{Chandrasekhar1975,Chandrasekhar1983}.
The mathematical structure behind the Schwarzschild transformations has
been studied from several points of view
\cite{AndersonPrice1991,Glampedakis2017}.
For Kerr perturbations, the Sasaki--Nakamura equation was introduced to
replace the long-range Teukolsky radial equation by an equation with
better asymptotic properties
\cite{SasakiNakamura1982}.
Other useful approaches include analytic series representations of
Teukolsky solutions and later studies of Sasaki--Nakamura-type
transformations
\cite{ManoSuzukiTakasugi1996,Nakano2016}.
Whiting used related transformation ideas in his proof of mode stability
for the Kerr black hole
\cite{Whiting1989}.

The present problem is more general.  We are not asking when two
Schr\"odinger operators are Darboux partners.  Instead, we start with a
given physical observable $T=A+B\partial_x$ and ask what happens to the
differential equation, its boundary solutions, and its spectral data
after applying this map.

The first issue is already visible at the local level.  If
\begin{equation}
 \begin{pmatrix}
 z\\ z'
 \end{pmatrix}
 =
 M(x,\omega)
 \begin{pmatrix}
 y\\ y'
 \end{pmatrix},
\end{equation}
then the determinant of the first-jet transformation is
\begin{equation}
 D(x,\omega)
 =
 A(A+B'-Bp)-B(A'-Bq).
 \label{eq:introD}
\end{equation}
Away from $D=0$, the transformation can be inverted locally and $z$
satisfies another second-order scalar equation.  At a zero of $D$, the
coefficients of this transformed scalar equation can become singular
even when the original equation is regular.

This phenomenon is not new by itself.  Equations satisfied by
derivatives of Heun and confluent-Heun functions can contain an
additional singular point
\cite{Shahnazaryan2012,Filipuk2020}.
More generally, apparent singularities arise when a first-order system
is reduced to a scalar equation by choosing a cyclic vector.  This is a
standard part of the theory of linear differential equations,
isomonodromic systems, and Painlev\'e equations
\cite{DubrovinMazzocco2007,vanderPutSaito2009}.
For this reason, the appearance of an additional singular point after
differentiation is not claimed as a new result here.

A recent black-hole example gives a direct motivation for the present
work.  Aly and Stojkovic studied the radial Teukolsky equation in
horizon-penetrating Kerr coordinates and found that the differential
equation for the first radial derivative contains an additional regular
singular point
\cite{AlyStojkovic2023}.
Their reconstructed quantities depend on both the radial solution and
its derivative.  This raises a simple question: does a singular point of
the equation satisfied by the derivative represent a genuine physical
or spectral singularity, or is it only a feature of the transformed
scalar equation?

To answer this question, one must distinguish several different
objects.  A singular point of an ordinary differential equation is not
the same as a singularity of one of its solutions.  A turning point is
not automatically a pole of a Green function.  A zero of a
position-dependent Wronskian is not automatically a quasinormal-mode
condition.  Most importantly for this paper, a zero
\begin{equation}
 D(x_0,\omega_0)=0
\end{equation}
at one value of the radial coordinate is not obviously related to a
global spectral condition in the frequency variable.

This distinction becomes especially clear for the derivative map
\begin{equation}
 T=\partial_x.
\end{equation}
In this case
\begin{equation}
 D=q.
\end{equation}
For a Schr\"odinger-type equation
\begin{equation}
 \Psi''+Q(x,\omega)\Psi=0,
\end{equation}
we have
\begin{equation}
 D=Q.
\end{equation}
The ordinary turning points
\begin{equation}
 Q(x_0,\omega)=0
\end{equation}
therefore become singular points of the scalar equation for
$X=\Psi'$.  The local question is whether these singularities are
apparent.  The global question, which is more important here, is
whether they can change the resonance or quasinormal-mode spectrum.

The relation with Darboux theory must also be treated carefully.
Darboux transformations go back to the classical theory of linear
differential equations, and their spectral properties were developed
further in the Sturm--Liouville and supersymmetric quantum-mechanics
literature
\cite{Darboux1882,Crum1955,Sukumar1985,
CooperKhareSukhatme1995}.
For a true Darboux transformation, the transformation operator is chosen
so that the two differential operators satisfy an intertwining
relation.  This extra condition gives strong relations between their
solution spaces and spectra.

Open systems require additional care because their boundary conditions
are not the usual square-integrable boundary conditions.  Supersymmetric
transformations of quasinormal and outgoing-wave problems were studied
in Refs.~\cite{Leung1999,Leung2001}, including the behavior of
exceptional modes.  The algebraically special Schwarzschild frequency
also gives an important example where the transformation itself becomes
special
\cite{Maassen2000}.  Green functions of genuine Darboux or
supersymmetric partner operators obey corresponding transformation
relations
\cite{SamsonovSukumarPupasov2005,SamsonovPupasov2005}.

These results do not imply that an arbitrary map of the form
$A+B\partial_x$ preserves a quasinormal-mode problem.  In particular,
Yurov and Yurov showed that generalized first-order Darboux-type
transformations need not be isospectral when the transformed problem is
assigned boundary conditions independently of the original ones
\cite{Yurov2019}.  This point is important for the present paper.  We
will not assume that a differential substitution automatically
preserves the spectrum.  Instead, we follow the actual images of the
left and right physical boundary solutions under $T$.

Quasinormal modes are global objects.  Early studies of Schwarzschild
quasinormal modes already showed that they are fixed by conditions at
both the horizon and infinity
\cite{ChandrasekharDetweiler1975}.
Leaver later developed analytic and continued-fraction methods for Kerr
and Schwarzschild perturbations
\cite{Leaver1985,Leaver1986}.
The mathematical and physical theory of black-hole quasinormal modes
has since become a large subject
\cite{KokkotasSchmidt1999,BertiCardosoStarinets2009,
KonoplyaZhidenko2011}.
For our purpose, the important point is that a quasinormal frequency is
determined by a global boundary-value problem, not by the vanishing of a
coefficient at one radial point.

The same idea appears in the language of Jost and Evans functions.  An
Evans function is an analytic determinant constructed from the
solution subspaces selected by the boundary conditions, and its zeros
locate spectral values \cite{AlexanderGardnerJones1990}.  Its relation
with Jost functions and Fredholm determinants has also been studied in
detail \cite{GesztesyLatushkinMakarov2007}.
This language is useful here because it separates a global spectral
zero from a local loss of rank in the transformation matrix.

There is also a rigorous resonance theory for one-dimensional
Schr\"odinger operators and black-hole wave equations.  Scattering
resonances on the real line can be described through meromorphic
continuation of the outgoing resolvent
\cite{Zworski1987}.  For black holes, resonance distributions and
resolvent methods were developed for spherically symmetric backgrounds
and later for rotating black holes
\cite{SaBarretoZworski1997,BonyHafner2008,
Dyatlov2011,Dyatlov2012}.
Related microlocal work gives a precise spectral meaning to
quasinormal modes in Kerr--de Sitter and other black-hole settings
\cite{Vasy2013}.
These results provide the global spectral framework used in the present
paper.

The aim of this work is to connect this global theory with the local
transformation
\begin{equation}
 T=A+B\partial_x
\end{equation}
without assuming a Darboux intertwining relation.  The central question
is whether an isolated radial zero of $D(x,\omega)$ can produce a new
spectral pole.

Our results have three main parts.  First, we derive the transformed
second-order equation and show that the determinant $D$ controls both
the transformed Wronskian and the local invertibility of the first-jet
map.  An isolated zero of $D$ can make the transformed scalar equation
singular, but a genuine loss of a homogeneous solution occurs only when
$D(\cdot,\omega)$ vanishes identically.

Second, we study the global boundary-value problem.  When the physical
left and right boundary lines are carried through the transformation,
the transformed Evans determinant differs from the original one only by
a nonzero holomorphic factor.  Isolated radial zeros of $D$ therefore do
not create resonances or quasinormal-mode frequencies.  We then show
that the pushed-forward response $T(\omega)R_L(\omega)$ has the same pole
orders as the original outgoing inverse away from global transformation
degeneracy.

Third, we compare the physical Green response with the unit-source
Green function of the transformed scalar equation.  The latter can
contain an explicit $1/D$ factor.  We prove an exact identity showing
that this apparent singular factor cancels once the source is
transformed correctly.  We then apply the general results to the
Regge--Wheeler and Zerilli equations and return briefly to the
Teukolsky problem that motivated the analysis.

The paper is organized as follows.  Section~\ref{sec:general} develops
the first-order transformation and its determinant.  Section~\ref{sec:local}
studies the local singularities produced by zeros of $D$.
Section~\ref{sec:global} gives the global spectral theorem, and
Sec.~\ref{sec:resolvent} gives the resolvent and Green-function results.
The remaining sections classify the exceptional cases, apply the theory
to Schwarzschild perturbations, compare it with Darboux and
Chandrasekhar transformations, and discuss the Teukolsky application.
Longer algebraic checks are kept in the appendices.


\section{First-order differential observables}
\label{sec:general}

We begin with the second-order equation
\begin{equation}
    L(\omega)y
    \equiv
    y''+p(x,\omega)y'+q(x,\omega)y
    =0.
    \label{eq:original}
\end{equation}
The independent variable is denoted by $x$, and $\omega$ is the
spectral parameter.  For the local calculations in this section,
the dependence on $\omega$ will usually be left implicit.

We study a first-order differential observable
\begin{equation}
    z=T y
    =
    A(x)y+B(x)y'.
    \label{eq:Tdef}
\end{equation}
We do not assume that $T$ is a Darboux transformation.
In particular, we do not assume an intertwining relation of the form
$\widetilde L T=T L$.

Define
\begin{equation}
    C=A'-Bq,
    \qquad
    E=A+B'-Bp.
    \label{eq:CE}
\end{equation}
Using the original equation to replace $y''$, we obtain
\begin{equation}
    z'=Cy+Ey'.
\end{equation}
Therefore
\begin{equation}
    \begin{pmatrix}
        z\\
        z'
    \end{pmatrix}
    =
    M(x)
    \begin{pmatrix}
        y\\
        y'
    \end{pmatrix},
    \qquad
    M(x)=
    \begin{pmatrix}
        A&B\\
        C&E
    \end{pmatrix}.
    \label{eq:matrixmap}
\end{equation}

The determinant of this first-jet map is
\begin{align}
    D(x)
    &=
    \det M
    \nonumber\\
    &=
    A(A+B'-Bp)-B(A'-Bq)
    \nonumber\\
    &=
    A^2+AB'-ABp-BA'+B^2q.
    \label{eq:Ddef}
\end{align}
The quantity $D$ will play two different roles.
Locally, it decides whether the first jet $(y,y')$ can be recovered
from $(z,z')$.  Globally, it will also tell us when $T$ loses an
entire homogeneous solution.

For later use define
\begin{equation}
    F=C'-Eq,
    \qquad
    G=C+E'-Ep.
    \label{eq:FG}
\end{equation}
Then
\begin{equation}
    z''=Fy+Gy'.
\end{equation}

\subsection{The transformed second-order equation}

\begin{theorem}[Transformed equation]
\label{thm:transformed}
On every open set on which $D\neq0$, the function $z=Ty$ satisfies
\begin{equation}
    \widetilde L z
    \equiv
    z''+\widetilde p\,z'+\widetilde q\,z
    =0,
    \label{eq:tildeL}
\end{equation}
where
\begin{equation}
    \boxed{
    \widetilde p
    =
    \frac{FB-GA}{D}
    =
    p-\frac{D'}{D}}
    \label{eq:ptilde}
\end{equation}
and
\begin{equation}
    \boxed{
    \widetilde q
    =
    \frac{GC-FE}{D}.}
    \label{eq:qtilde}
\end{equation}
The inverse first-jet relation is
\begin{equation}
    \begin{pmatrix}
        y\\
        y'
    \end{pmatrix}
    =
    \frac{1}{D}
    \begin{pmatrix}
        E&-B\\
        -C&A
    \end{pmatrix}
    \begin{pmatrix}
        z\\
        z'
    \end{pmatrix}.
    \label{eq:inverse}
\end{equation}
\end{theorem}

\begin{proof}
From \eqref{eq:matrixmap},
\begin{equation}
    y=\frac{Ez-Bz'}{D},
    \qquad
    y'=\frac{-Cz+Az'}{D}.
\end{equation}
Substituting these expressions into
$z''=Fy+Gy'$ gives
\begin{equation}
    z''
    =
    \frac{FE-GC}{D}z
    +
    \frac{-FB+GA}{D}z'.
\end{equation}
This proves the transformed equation.
A direct differentiation of $D$ gives
\begin{equation}
    FB-GA=pD-D',
\end{equation}
which gives the second form in \eqref{eq:ptilde}.
\end{proof}

The poles in \eqref{eq:ptilde} and \eqref{eq:qtilde} show why zeros of
$D$ matter.  They can make the scalar equation for $z$ singular even
when the original equation for $y$ is regular.

\subsection{Wronskian transformation}

\begin{theorem}[Wronskian identity]
\label{thm:wronskian}
For any two homogeneous solutions $y_1,y_2$,
\begin{equation}
    \boxed{
    W[Ty_1,Ty_2]
    =
    D(x)W[y_1,y_2].}
    \label{eq:Wlaw}
\end{equation}
\end{theorem}

\begin{proof}
Apply \eqref{eq:matrixmap} to the two solution columns and take the
determinant:
\begin{equation}
    \det
    \begin{pmatrix}
        Ty_1&Ty_2\\
        (Ty_1)'&(Ty_2)'
    \end{pmatrix}
    =
    \det M
    \det
    \begin{pmatrix}
        y_1&y_2\\
        y_1'&y_2'
    \end{pmatrix}.
\end{equation}
This is exactly \eqref{eq:Wlaw}.
\end{proof}

Abel's identity gives a useful check.  Since
\begin{equation}
    W_y'=-pW_y,
    \qquad
    W_z'=-\widetilde pW_z,
\end{equation}
and $W_z=DW_y$, one again finds
\begin{equation}
    \widetilde p=p-\frac{D'}{D}.
\end{equation}

\subsection{An isolated zero and a global kernel}

An isolated zero of $D$ makes the local inverse
\eqref{eq:inverse} singular.  It does not necessarily mean that
the transformation loses one of the two global solutions.

This distinction is important for the spectral problem.

\begin{theorem}[Global kernel criterion]
\label{thm:global-kernel}
Let the original equation be regular on a connected interval or
simply connected domain.  Then
\begin{equation}
    \boxed{
    \ker\!
    \left(
       T|_{\ker L}
    \right)
    \neq\{0\}
    \quad\Longleftrightarrow\quad
    D(\cdot,\omega)\equiv0.}
    \label{eq:kerneliff}
\end{equation}
Therefore an isolated zero
\begin{equation}
    D(x_0,\omega)=0,
    \qquad
    D(\cdot,\omega)\not\equiv0,
\end{equation}
is only a local loss of invertibility of the first-jet map.
\end{theorem}

\begin{proof}
Suppose first that a nonzero homogeneous solution $y_1$ satisfies
$Ty_1=0$.  Choose an independent homogeneous solution $y_2$.
The Wronskian identity gives
\begin{equation}
    0
    =
    W[Ty_1,Ty_2]
    =
    D\,W[y_1,y_2].
\end{equation}
The original Wronskian is nonzero on the connected regular domain.
Hence
\begin{equation}
    D\equiv0.
\end{equation}

Conversely, suppose $D\equiv0$.  For a basis $y_1,y_2$ of the original
solution space,
\begin{equation}
    W[Ty_1,Ty_2]\equiv0.
\end{equation}
The two transformed functions are therefore linearly dependent.
Hence some nonzero constant linear combination
$c_1y_1+c_2y_2$ is killed by $T$.
\end{proof}

This result will be used again in the resolvent section.
It shows that the condition
\begin{equation}
    D(x_0,\omega_0)=0
\end{equation}
at one point is very different from
\begin{equation}
    D(\cdot,\omega_0)\equiv0.
\end{equation}

\subsection{Source transformation}

The inhomogeneous equation is
\begin{equation}
    Ly=S.
    \label{eq:inhom-original}
\end{equation}
For Green functions, it is not enough to transform only the
homogeneous solutions.  The source must also be transformed.

\begin{theorem}[Source identity]
\label{thm:sourcepush}
On every open set where $D\neq0$,
\begin{equation}
    \boxed{
    \widetilde L T
    =
    \mathcal S L,}
    \label{eq:operatoridentity}
\end{equation}
where
\begin{equation}
    \boxed{
    \mathcal S
    =
    B\partial_x
    +
    A+2B'
    -
    B\frac{D'}{D}.}
    \label{eq:Sdef}
\end{equation}
Thus, if $Ly=S$ and $z=Ty$,
\begin{equation}
    \boxed{
    \widetilde Lz
    =
    BS'
    +
    \left(
      A+2B'
      -
      B\frac{D'}{D}
    \right)S.}
    \label{eq:sourcepush}
\end{equation}
\end{theorem}

\begin{proof}
Write
\begin{equation}
    T=B\partial_x+A,
    \qquad
    \mathcal S=B\partial_x+H.
\end{equation}
The coefficient of $\partial_x^3$ agrees automatically.
Matching the coefficient of $\partial_x^2$ gives
\begin{equation}
    H
    =
    A+2B'+\widetilde p B-Bp.
\end{equation}
Using \eqref{eq:ptilde},
\begin{equation}
    H
    =
    A+2B'
    -
    B\frac{D'}{D}.
\end{equation}
The remaining coefficients then reduce to
\eqref{eq:qtilde}.  This proves the operator identity.
\end{proof}

Equation \eqref{eq:operatoridentity} is weaker than an ordinary
Darboux intertwining relation.  In general,
\begin{equation}
    \mathcal S\neq T.
\end{equation}

\subsection{A differential quasi-inverse}

The same determinant gives a useful inverse modulo the equations.
Define
\begin{equation}
    U
    =
    \frac{E-B\partial_x}{D},
    \qquad
    E=A+B'-Bp.
    \label{eq:Udef}
\end{equation}

\begin{theorem}[B\'ezout-type identities]
\label{thm:bezout}
Where $D\neq0$,
\begin{equation}
    \boxed{
    UT
    =
    I-\frac{B^2}{D}L,}
    \label{eq:UT}
\end{equation}
and
\begin{equation}
    \boxed{
    TU
    =
    I-\frac{B^2}{D}\widetilde L.}
    \label{eq:TU}
\end{equation}
Hence $T$ and $U$ are inverse maps on the local homogeneous
solution spaces.
\end{theorem}

\begin{proof}
A direct composition gives
\begin{equation}
    (E-B\partial_x)(A+B\partial_x)
    =
    D-B^2L.
\end{equation}
Dividing by $D$ gives \eqref{eq:UT}.

For the second identity, $TU-I$ is a second-order operator whose
leading coefficient is $-B^2/D$.  It annihilates every local solution
of $\widetilde Lz=0$.  Since $\widetilde L$ is monic,
\begin{equation}
    TU-I
    =
    -\frac{B^2}{D}\widetilde L.
\end{equation}
\end{proof}

For an inhomogeneous solution,
\begin{equation}
    Ly=f,
    \qquad
    z=Ty,
\end{equation}
Eq.~\eqref{eq:UT} gives
\begin{equation}
    \boxed{
    y
    =
    Uz+\frac{B^2}{D}f.}
    \label{eq:inhomreconstruct}
\end{equation}
The separate factors in this formula may be singular at $D=0$.
For a physical pair $(z,f)$ obtained from a regular original solution,
the complete expression has a regular continuation.

A strict intertwining relation
\begin{equation}
    \widetilde L T=T L
\end{equation}
requires
\begin{equation}
    2B'
    -
    B\frac{D'}{D}
    =0.
\end{equation}
For $B\neq0$ this is equivalent to
\begin{equation}
    \boxed{
    \left(\frac{D}{B^2}\right)'=0.}
    \label{eq:strictcriterion}
\end{equation}
This already shows that a general differential observable is a larger
class than a Darboux transformation
\cite{Darboux1882,Crum1955,Sukumar1985,CooperKhareSukhatme1995,
Glampedakis2017}.


\section{Determinant zeros and apparent singularities}
\label{sec:local}

We now study what happens when $D$ vanishes at an ordinary point of
the original equation.  Apparent singularities produced by scalar
reductions and cyclic vectors are well known in the theory of linear
differential equations
\cite{DubrovinMazzocco2007,vanderPutSaito2009}.
Related extra singularities also occur in equations satisfied by
derivatives of Heun functions
\cite{Shahnazaryan2012,Filipuk2020}.
The purpose here is to give the form needed later for the global
spectral problem.

Let
\begin{equation}
    t=x-x_0.
\end{equation}
Assume that $p,q,A,B$ are holomorphic at $x_0$ and that
$x_0$ is an ordinary point of the original equation.

We use the following definition.

\begin{definition}
A regular singular point of the transformed scalar equation is called
\emph{apparent} if it has a fundamental pair of solutions that are
single-valued and holomorphic at the point.  In this case no logarithmic
solution is present and the local monodromy is the identity.
\end{definition}

Write
\begin{equation}
    A=t^k a(t),
    \qquad
    B=t^k b(t),
    \label{eq:factorAB}
\end{equation}
where
\begin{equation}
    k=\min\{
        \operatorname{ord}_{x_0}A,
        \operatorname{ord}_{x_0}B
    \},
\end{equation}
and $a,b$ are holomorphic and not both zero at $t=0$.

The determinant has the form
\begin{equation}
    D=t^{2k}\widehat D.
\end{equation}
Let
\begin{equation}
    m=\operatorname{ord}_{x_0}D.
\end{equation}

\begin{theorem}[Local determinant zero]
\label{thm:localdeg}
Suppose
\begin{equation}
    D(x_0)=0,
    \qquad
    D\not\equiv0.
\end{equation}
Then $x_0$ is a regular singular point of the transformed scalar
equation.  Its Frobenius exponents are
\begin{equation}
    \boxed{
    \rho_1=k,
    \qquad
    \rho_2=m-k+1.}
    \label{eq:localexponents}
\end{equation}
The two physical transformed solutions are holomorphic, no logarithmic
solution occurs, and the local monodromy is trivial.
\end{theorem}

\begin{proof}
Remove the common factor $t^k$ and define
\begin{equation}
    \widehat z
    =
    ay+by'.
\end{equation}
Let
\begin{equation}
    \mu
    =
    \operatorname{ord}_0\widehat D
    =
    m-2k.
\end{equation}

Since $(a(0),b(0))\neq(0,0)$, the number
$\widehat z(0)$ is a nonzero linear functional on the two-dimensional
Cauchy-data space of the original equation.  We may therefore choose a
homogeneous basis $y_1,y_2$ such that
\begin{equation}
    \widehat z_1(0)\neq0,
    \qquad
    \widehat z_2(0)=0.
\end{equation}

The Wronskian identity gives
\begin{equation}
    W[\widehat z_1,\widehat z_2]
    =
    \widehat D
    W[y_1,y_2].
\end{equation}
The original Wronskian does not vanish at an ordinary point.
Hence
\begin{equation}
    \operatorname{ord}_0
    W[\widehat z_1,\widehat z_2]
    =
    \mu.
\end{equation}
If
\begin{equation}
    s=\operatorname{ord}_0\widehat z_2,
\end{equation}
then $\widehat z_1(0)\neq0$ gives
\begin{equation}
    \operatorname{ord}_0
    W[\widehat z_1,\widehat z_2]
    =
    s-1.
\end{equation}
Therefore
\begin{equation}
    s=\mu+1.
\end{equation}

Multiplying again by $t^k$, the two physical transformed solutions
have orders
\begin{equation}
    k,
    \qquad
    k+\mu+1
    =
    m-k+1.
\end{equation}

From \eqref{eq:ptilde},
\begin{equation}
    \widetilde p
    =
    -\frac{m}{t}+O(1).
\end{equation}
The two analytic solution orders above give the indicial polynomial
\begin{equation}
    (\rho-k)
    (\rho-m+k-1)
    =0.
\end{equation}
Thus the point is regular singular.  Since the two independent
physical solutions are already holomorphic, no logarithmic branch is
present and the monodromy is trivial.
\end{proof}

If $A$ and $B$ do not share a zero, then $k=0$ and
\begin{equation}
    \boxed{
    \rho_1=0,
    \qquad
    \rho_2=m+1.}
    \label{eq:kzeroexp}
\end{equation}
A simple zero of $D$ therefore gives the exponents
\begin{equation}
    0,\qquad2.
\end{equation}

A zero of $B$ alone is not special.
If $A(x_0)\neq0$, then $k=0$, and the singularity is controlled only
by the determinant.

If $A$ and $B$ are meromorphic rather than holomorphic, the same
argument works after removing their common meromorphic factor.
The local exponents can then be shifted by an integer.  The physical
solutions may contain the same common pole, but the induced local
monodromy remains trivial.

If a zero of $D$ coincides with an original singular point, this
ordinary-point theorem does not apply.  Horizons and infinity are
examples of this situation in black-hole problems.  They must be
studied using their own asymptotic solution spaces.

\subsection{The derivative observable}

Take
\begin{equation}
    T=\partial_x.
\end{equation}
Then
\begin{equation}
    A=0,
    \qquad
    B=1,
    \qquad
    D=q.
\end{equation}
For
\begin{equation}
    z=y',
\end{equation}
the transformed equation is
\begin{equation}
    \boxed{
    z''
    +
    \left(
       p-\frac{q'}{q}
    \right)z'
    +
    \left(
       p'+q-p\frac{q'}{q}
    \right)z
    =0.}
    \label{eq:derivativeeq}
\end{equation}

This derivative equation and its extra singular points are known in
the Heun literature
\cite{Filipuk2020}.
Our later results will concern their global spectral meaning.

\begin{corollary}
\label{cor:derivativezero}
If $q$ has a zero of multiplicity $m\ge1$ at an ordinary point of the
original equation, then the derivative equation has an apparent regular
singularity with exponents
\begin{equation}
    \boxed{
    0,\qquad m+1.}
\end{equation}
\end{corollary}

The singular coefficients of \eqref{eq:derivativeeq} do not mean that
$z=y'$ is singular.  The original function $y$ is analytic at the
ordinary point, and therefore $z$ is analytic there as well.

There is only a loss of first-jet coordinates at that point.

\subsection{Recovery of the missing datum}

The missing original datum can be recovered from a higher derivative
of the transformed solution.

Suppose
\begin{equation}
    q=t^m h(t),
    \qquad
    h(0)\neq0.
\end{equation}
Define
\begin{equation}
    H_z=z'+pz.
\end{equation}
Using $z=y'$ and the original equation gives
\begin{equation}
    H_z=-qy.
    \label{eq:Hz}
\end{equation}
Thus $H_z$ vanishes to order at least $m$, and
\begin{equation}
    \boxed{
    y(x_0)
    =
    -\frac{H_z^{(m)}(x_0)}
    {q^{(m)}(x_0)},}
    \label{eq:jetrecover1}
\end{equation}
while
\begin{equation}
    \boxed{
    y'(x_0)=z(x_0).}
    \label{eq:jetrecover2}
\end{equation}

For a simple zero,
\begin{equation}
    y(x_0)
    =
    \frac{
      (p^2-p')z-z''
    }{q'}
    \bigg|_{x=x_0}.
\end{equation}
The transformation therefore has not lost the solution.
One datum has only moved from the first jet of $z$ to a higher jet.

\subsection{Schr\"odinger equations and turning points}

Now take
\begin{equation}
    \Psi''+Q(x,\omega)\Psi=0,
    \qquad
    Q=\omega^2-V(x).
    \label{eq:schrodinger}
\end{equation}
For
\begin{equation}
    X=\Psi',
\end{equation}
Eq.~\eqref{eq:derivativeeq} becomes
\begin{equation}
    \boxed{
    X''
    -
    \frac{Q'}{Q}X'
    +
    QX
    =0.}
    \label{eq:Xeq}
\end{equation}

An ordinary turning point of the original equation satisfies
\begin{equation}
    Q(x_0,\omega)=0.
\end{equation}
It is therefore a determinant zero for the derivative map.

\begin{theorem}[Turning-point singularity]
\label{thm:turning}
Suppose $Q$ has a zero of multiplicity $m$ at an ordinary point
$x_0$.  Then the equation for $X=\Psi'$ has an apparent regular
singularity at $x_0$ with exponents
\begin{equation}
    \boxed{
    0,\qquad m+1.}
\end{equation}
The original data are recovered from
\begin{equation}
    \boxed{
    \Psi(x_0)
    =
    -\frac{X^{(m+1)}(x_0)}
    {Q^{(m)}(x_0)},
    \qquad
    \Psi'(x_0)=X(x_0).}
    \label{eq:turningrecover}
\end{equation}
\end{theorem}

For a simple turning point,
\begin{equation}
    X'(x_0)=0,
\end{equation}
and
\begin{equation}
    \Psi(x_0)
    =
    -\frac{X''(x_0)}{Q'(x_0)}.
\end{equation}

The Wronskian relation is
\begin{equation}
    \boxed{
    W[X_1,X_2]
    =
    Q(x,\omega)
    W[\Psi_1,\Psi_2].}
    \label{eq:turningW}
\end{equation}
Thus
\begin{equation}
    W[X_1,X_2](x_0)=0
\end{equation}
at a turning point.
This does not mean that the two transformed solutions are globally
dependent.  The turning point is singular for the equation
\eqref{eq:Xeq}, so the usual ordinary-point Wronskian argument cannot
be applied there.

The same point explains why an apparent singularity does not
automatically limit the Taylor series of a physical transformed
solution.  For example,
\begin{equation}
    y''+xy=0
\end{equation}
has entire solutions and therefore entire derivatives, while
$z=y'$ satisfies
\begin{equation}
    z''-\frac1x z'+xz=0.
\end{equation}
The coefficient singularity at $x=0$ is apparent and does not create
a singularity of the derivative itself.


\section{Global boundary problem and spectral determinant}
\label{sec:global}

The local results above do not yet define a spectrum.
A resonance or quasinormal mode depends on boundary conditions at
both ends of the radial interval.

We therefore let
\begin{equation}
    L(\omega)
    =
    \partial_x^2
    +
    p(x,\omega)\partial_x
    +
    q(x,\omega)
    \label{eq:Lomega}
\end{equation}
act on
\begin{equation}
    I=(a,b).
\end{equation}
For black-hole radial equations one usually has
\begin{equation}
    a=-\infty,
    \qquad
    b=+\infty
\end{equation}
in tortoise coordinate.

We work on a connected frequency domain $\Omega$ on a fixed sheet of
the spectral surface.  Thresholds and branch points are excluded from
$\Omega$ unless they are treated separately.

Let
\begin{equation}
    \ell_-(\omega)
    =
    \operatorname{span}\{u_-(\cdot,\omega)\},
\end{equation}
and
\begin{equation}
    \ell_+(\omega)
    =
    \operatorname{span}\{u_+(\cdot,\omega)\}
\end{equation}
be the one-dimensional solution spaces selected by the left and
right physical boundary conditions.

For a short-range black-hole problem these are the horizon-ingoing
and infinity-outgoing solution lines.
This is the usual Jost-function point of view used in scattering and
quasinormal-mode theory
\cite{ChandrasekharDetweiler1975,Leaver1985,Leaver1986,
KokkotasSchmidt1999,BertiCardosoStarinets2009,
KonoplyaZhidenko2011}.

When $p\neq0$, the raw Wronskian is not constant.
Choose an ordinary base point $x_*$ and define
\begin{equation}
    \rho_L(x,\omega)
    =
    \exp
    \left(
       \int_{x_*}^{x}
       p(s,\omega)\,\dd s
    \right).
    \label{eq:rhoL}
\end{equation}
The Abel-normalized determinant
\begin{equation}
    \boxed{
    \mathcal E_L(\omega)
    =
    \rho_L
    W[u_-,u_+]}
    \label{eq:EL}
\end{equation}
is independent of $x$.

Changing the normalization of either boundary solution multiplies
$\mathcal E_L$ by a nonzero holomorphic function.
Its zeros, with their orders, are therefore the relevant spectral data.
This is the usual freedom in Evans- and Jost-function constructions
\cite{AlexanderGardnerJones1990,
GesztesyLatushkinMakarov2007}.

\subsection{How the boundary solutions transform}

The natural transformed boundary lines are
\begin{equation}
    T\ell_-(\omega),
    \qquad
    T\ell_+(\omega).
    \label{eq:imagelines}
\end{equation}
The important point is that the boundary conditions must be transported
with the observable.  One cannot in general transform the differential
equation and then choose unrelated boundary conditions for the new
equation.

Suppose near one endpoint
\begin{equation}
    u(x,\omega)
    =
    e^{\lambda(\omega)x}
    \left[
      c(\omega)+o(1)
    \right],
\end{equation}
with
\begin{equation}
    \frac{u'}{u}
    \longrightarrow
    \lambda(\omega).
\end{equation}
If
\begin{equation}
    A\to A_\infty(\omega),
    \qquad
    B\to B_\infty(\omega),
\end{equation}
then
\begin{equation}
    Tu
    =
    u
    \left[
       N(\omega)+o(1)
    \right],
\end{equation}
where
\begin{equation}
    \boxed{
    N(\omega)
    =
    A_\infty(\omega)
    +
    B_\infty(\omega)\lambda(\omega).}
    \label{eq:Nfactor}
\end{equation}

If $N\neq0$, the leading asymptotic type is unchanged.
If $N=0$, one must compute the next asymptotic term.
A zero of $N$ does not automatically mean
\begin{equation}
    Tu\equiv0.
\end{equation}

This distinction is needed below.

\subsection{The transformed Evans determinant}

The transformed first-derivative coefficient is
\begin{equation}
    \widetilde p
    =
    p-\frac{D'}{D}.
\end{equation}
Hence a compatible Abel factor is
\begin{equation}
    \rho_T(x,\omega)
    =
    \rho_L(x,\omega)
    \frac{D(x_*,\omega)}
         {D(x,\omega)}.
    \label{eq:rhoT}
\end{equation}

The apparent factor $1/D$ is cancelled by the Wronskian identity.

\begin{theorem}[Global spectral determinant]
\label{thm:evans}
Fix $\omega_0\in\Omega$.  Assume that

\begin{enumerate}
    \item $p,q,A,B$ and the two chosen boundary solutions depend
    holomorphically on $\omega$ near $\omega_0$;

    \item
    \begin{equation}
        D(\cdot,\omega_0)\not\equiv0;
    \end{equation}

    \item every interior zero of $D$ in the region considered is
    isolated and occurs at an ordinary point of the original equation;

    \item the corresponding transformed physical solutions extend
    through those zeros;

    \item $T\ell_-$ and $T\ell_+$ are nonzero and are used as the
    boundary lines of the transformed observable problem;

    \item no determinant zero reaches an endpoint or an original
    singular point in the frequency neighbourhood considered.
\end{enumerate}

Choose $x_*$ such that
\begin{equation}
    D(x_*,\omega_0)\neq0.
\end{equation}
Then, after reducing the frequency neighbourhood if necessary,
\begin{equation}
    \boxed{
    \mathcal E_T^{\rm img}(\omega)
    =
    D(x_*,\omega)\,
    \mathcal E_L(\omega).}
    \label{eq:EvansImage}
\end{equation}
Therefore the two spectral determinants have the same zeros with the
same orders.

In particular,
\begin{equation}
    D(x_0,\omega_0)=0,
    \qquad
    D(\cdot,\omega_0)\not\equiv0,
\end{equation}
does not create a resonance or a quasinormal-mode frequency.
\end{theorem}

\begin{proof}
Using \eqref{eq:rhoT} and Theorem~\ref{thm:wronskian},
\begin{align}
    \mathcal E_T^{\rm img}
    &=
    \rho_T
    W[Tu_-,Tu_+]
    \nonumber\\
    &=
    \rho_L
    \frac{D(x_*)}{D(x)}
    D(x)
    W[u_-,u_+]
    \nonumber\\
    &=
    D(x_*)\mathcal E_L.
\end{align}
Since
\begin{equation}
    D(x_*,\omega_0)\neq0,
\end{equation}
the factor $D(x_*,\omega)$ is holomorphic and nonzero in a sufficiently
small neighbourhood of $\omega_0$.  It cannot create or remove a zero
of the Evans determinant, and it cannot change the order of a zero.
\end{proof}

If the transformed solutions are renormalized to unit Jost amplitudes
and the endpoint factors $N_-$ and $N_+$ are nonzero, then
\begin{equation}
    \boxed{
    \widehat{\mathcal E}_T
    =
    \frac{
       D(x_*,\omega)
    }{
       N_-(\omega)N_+(\omega)
    }
    \mathcal E_L.}
    \label{eq:EvansNormalized}
\end{equation}

Equation~\eqref{eq:EvansImage} is the cleaner form because it follows
the physical image boundary lines directly.

The theorem preserves the order of an Evans/Jost zero.
In settings where the chosen Evans function is known to reproduce
algebraic multiplicity of an analytic Fredholm characteristic value,
the corresponding algebraic multiplicity is also preserved
\cite{GesztesyLatushkinMakarov2007}.
We will not assume that additional identification in the abstract
theorem.

The boundary-line hypothesis is essential.
Generalized first-order transformations are not automatically
quasinormal-mode isospectral when one imposes independent boundary
conditions on the transformed equation
\cite{Yurov2019}.
Our theorem concerns a physical observable and therefore follows the
actual image of the original boundary data.


\section{Pushed-forward resolvent and Green function}
\label{sec:resolvent}

We now pass from the spectral determinant to the inhomogeneous
operator problem.

Outgoing resolvents and resonances can be defined rigorously for
one-dimensional scattering problems and for several black-hole wave
equations
\cite{Zworski1987,SaBarretoZworski1997,BonyHafner2008,
Dyatlov2011,Dyatlov2012,Dyatlov2015,Vasy2013}.
For the present argument we only need a local analytic setup in the
frequency variable.

Let $X$ and $Y$ be Banach spaces and suppose
\begin{equation}
    L(\omega):Y\rightarrow X
\end{equation}
is a holomorphic Fredholm family of index zero on $\Omega$.
Assume that it is invertible at least at one point.
Analytic Fredholm theory then gives a meromorphic inverse
\begin{equation}
    R_L(\omega)
    =
    L(\omega)^{-1}.
    \label{eq:RL}
\end{equation}
In simple terms, $R_L$ depends regularly on $\omega$ except at isolated
poles inside the frequency region considered here.  We use $R_L$ for
the outgoing inverse of the chosen boundary-value problem.  Branch
points and threshold frequencies are treated separately.

Assume also that
\begin{equation}
    T(\omega):Y\rightarrow Z
\end{equation}
is a holomorphic bounded operator family.

Define the physical pushed-forward response
\begin{equation}
    \boxed{
    \mathcal G_T(\omega)
    =
    T(\omega)R_L(\omega).}
    \label{eq:GT}
\end{equation}

\begin{theorem}[Pushed-forward resolvent]
\label{thm:pushres}
The family $\mathcal G_T$ is meromorphic on the same frequency domain
as $R_L$.

If $R_L$ is holomorphic at $\omega_0$, then $\mathcal G_T$ is also
holomorphic there.

Suppose instead that $R_L$ has a pole of order $m$,
\begin{equation}
    R_L(\omega)
    =
    \sum_{j=-m}^{\infty}
    R_j(\omega-\omega_0)^j,
    \qquad
    R_{-m}\neq0.
    \label{eq:RLaurent}
\end{equation}
If
\begin{equation}
    D(\cdot,\omega_0)\not\equiv0,
\end{equation}
then $\mathcal G_T$ has a pole of the same order $m$.
\end{theorem}

\begin{proof}
Since $T$ is holomorphic in $\omega$, multiplying the meromorphic
family $R_L$ by $T$ cannot create a pole at a frequency where $R_L$ is
regular.

Now expand
\begin{equation}
    T(\omega)
    =
    T_0+T_1(\omega-\omega_0)+\cdots.
\end{equation}
The highest negative power in
\begin{equation}
    L(\omega)R_L(\omega)=I
\end{equation}
gives
\begin{equation}
    L(\omega_0)R_{-m}=0.
\end{equation}
Hence
\begin{equation}
    \operatorname{Ran}R_{-m}
    \subset
    \ker L(\omega_0).
\end{equation}

By Theorem~\ref{thm:global-kernel},
\begin{equation}
    D(\cdot,\omega_0)\not\equiv0
\end{equation}
implies that $T_0$ is injective on the homogeneous solution space.
Since $R_{-m}\neq0$,
\begin{equation}
    T_0R_{-m}\neq0.
\end{equation}
This is the coefficient of
$(\omega-\omega_0)^{-m}$ in $TR_L$.
The pole order is therefore unchanged.
\end{proof}

This theorem gives a stronger statement than the Evans-function result:
an isolated determinant zero cannot create a new physical resolvent
pole, and it cannot remove an existing pole either.

For a simple pole,
\begin{equation}
    R_L(\omega)
    =
    \frac{P_0}{\omega-\omega_0}
    +
    R_{\rm reg}(\omega),
\end{equation}
we have
\begin{equation}
    \boxed{
    \operatorname*{Res}_{\omega=\omega_0}
    \mathcal G_T
    =
    T(\omega_0)P_0.}
    \label{eq:residue}
\end{equation}
The derivative $\partial_\omega T$ enters only the regular term.

A simple resonance can disappear from the observable only if
\begin{equation}
    T(\omega_0)P_0=0.
\end{equation}
Since the range of $P_0$ consists of homogeneous resonant states,
Theorem~\ref{thm:global-kernel} then gives
\begin{equation}
    D(\cdot,\omega_0)\equiv0.
\end{equation}
Thus a mode can be invisible to the observable only at a genuine
global transformation degeneracy.

\subsection{The two Green kernels}

Let $u_-$ and $u_+$ be the two boundary-adapted homogeneous solutions.
The original Green kernel is
\begin{equation}
    G_L(x,s;\omega)
    =
    \frac{
      u_-(x_<,\omega)
      u_+(x_>,\omega)
    }{
      W_L(s,\omega)
    },
    \label{eq:Gorig}
\end{equation}
where
\begin{equation}
    x_< =\min(x,s),
    \qquad
    x_> =\max(x,s).
\end{equation}

The physical observable kernel is simply
\begin{equation}
    \boxed{
    G_T(x,s;\omega)
    =
    T_xG_L(x,s;\omega).}
    \label{eq:physicalG}
\end{equation}
No factor $1/D$ appears in this expression.

Now construct instead the unit-$\delta$ Green kernel of the scalar
transformed equation.  With
\begin{equation}
    z_\pm=Tu_\pm,
\end{equation}
Theorem~\ref{thm:wronskian} gives
\begin{equation}
    W[z_-,z_+](s)
    =
    D(s,\omega)W_L(s,\omega).
\end{equation}
Hence
\begin{equation}
    \boxed{
    G_{\widetilde L}^{(\delta)}(x,s;\omega)
    =
    \frac{
      z_-(x_<,\omega)
      z_+(x_>,\omega)
    }{
      D(s,\omega)W_L(s,\omega)
    }.}
    \label{eq:Gtildedelta}
\end{equation}

This kernel can have a genuine $1/D$ divergence.
It is therefore important not to identify
$G_{\widetilde L}^{(\delta)}$ with the physical kernel
$T_xG_L$.

The difference comes from the source transformation.

\subsection{Exact cancellation of the artificial pole}

The formal adjoint of $\mathcal S$ with respect to the unweighted
$\dd s$ pairing is
\begin{equation}
    \mathcal S_s^\dagger
    =
    -B\partial_s
    +
    A+B'
    -
    B\frac{D'}{D}.
    \label{eq:Sadj}
\end{equation}

\begin{theorem}[Exact Green identity]
\label{thm:GreenExact}
Away from zeros of $D$,
\begin{equation}
    \boxed{
    \mathcal S_s^\dagger
    G_{\widetilde L}^{(\delta)}(x,s;\omega)
    =
    T_xG_L(x,s;\omega).}
    \label{eq:GreenExact}
\end{equation}

If $s_0$ is an isolated apparent zero of $D$, the full left-hand side
has the unique removable continuation defined by the regular
right-hand side, even though
$G_{\widetilde L}^{(\delta)}$ and the separate terms in
$\mathcal S^\dagger G_{\widetilde L}^{(\delta)}$ may diverge.
\end{theorem}

\begin{proof}
Let $y$ be either $u_-$ or $u_+$ and let
\begin{equation}
    z=Ty.
\end{equation}
Since
\begin{equation}
    W_L'=-pW_L,
\end{equation}
a direct calculation gives
\begin{align}
    \mathcal S_s^\dagger
    \left(
      \frac{z}{DW_L}
    \right)
    &=
    \frac{
       -Bz'
       +(A+B'-Bp)z
    }{
       DW_L
    }
    \nonumber\\
    &=
    \frac{
       Ez-Bz'
    }{
       DW_L
    }.
\end{align}
The inverse first-jet relation gives
\begin{equation}
    \frac{Ez-Bz'}{D}=y.
\end{equation}
Therefore
\begin{equation}
    \boxed{
    \mathcal S_s^\dagger
    \left(
      \frac{z}{DW_L}
    \right)
    =
    \frac{y}{W_L}.}
    \label{eq:basicGreenid}
\end{equation}

Applying this relation to the $s$-dependent factor on the two sides of
$s=x$ gives \eqref{eq:GreenExact}.
The transformed Green kernel is continuous at $s=x$, so the
first-order adjoint operation does not create an additional delta
distribution there.
\end{proof}

Now let
\begin{equation}
    Ly=f.
\end{equation}
Then
\begin{equation}
    \widetilde L(Ty)
    =
    \mathcal Sf.
\end{equation}
For a smooth compactly supported source,
integration by parts gives
\begin{align}
    \int
    G_{\widetilde L}^{(\delta)}(x,s)
    (\mathcal Sf)(s)\,\dd s
    &=
    \int
    \mathcal S_s^\dagger
    G_{\widetilde L}^{(\delta)}(x,s)
    f(s)\,\dd s
    \nonumber\\
    &=
    T_x
    \int
    G_L(x,s)f(s)\,\dd s.
\end{align}
Thus
\begin{equation}
    \boxed{
    R_{\widetilde L}^{(\delta)}
    \mathcal S
    =
    TR_L}
    \label{eq:GreenCompose}
\end{equation}
on this source class.

For a noncompact source the boundary term
\begin{equation}
    \left[
      B(s)
      G_{\widetilde L}^{(\delta)}(x,s)
      f(s)
    \right]_{a}^{b}
\end{equation}
must be included unless the boundary conditions make it vanish.

For a delta source $f=\delta_{s_0}$ away from $D=0$,
\begin{equation}
    \left\langle
       G_{\widetilde L}^{(\delta)},
       \mathcal S\delta_{s_0}
    \right\rangle
    =
    \left.
    \mathcal S_s^\dagger
    G_{\widetilde L}^{(\delta)}
    \right|_{s=s_0}
    =
    T_xG_L(x,s_0).
\end{equation}
If $D(s_0)=0$, the correct object is the regular combined
continuation supplied by Theorem~\ref{thm:GreenExact}, rather than
term-by-term multiplication of singular distributions.

This Green identity is the main difference between the present
non-intertwining observable problem and the usual Green-function
relations for supersymmetric or Darboux partner operators
\cite{SamsonovSukumarPupasov2005,SamsonovPupasov2005}.


\section{Exceptional cases}
\label{sec:exceptions}

The previous theorems require several hypotheses.
It is useful to state clearly what happens when one of them fails.

\begin{table*}[t]
\caption{Main cases that must be distinguished in applying the spectral theorem.}
\label{tab:exceptions}
\begin{center}
\scriptsize
\begin{tabular}{llll}
\hline\hline
\parbox{0.21\textwidth}{Case} &
\parbox{0.22\textwidth}{Meaning} &
\parbox{0.23\textwidth}{Spectral effect} &
\parbox{0.25\textwidth}{What must be done} \\
\hline
\parbox{0.21\textwidth}{$\mathcal E_L(\omega_0)=0$} &
\parbox{0.22\textwidth}{Original resonance} &
\parbox{0.23\textwidth}{Genuine spectral zero and resolvent pole} &
\parbox{0.25\textwidth}{No special transformation issue} \\
\parbox{0.21\textwidth}{$D(x_0,\omega_0)=0$, $D\not\equiv0$} &
\parbox{0.22\textwidth}{Isolated local degeneracy} &
\parbox{0.23\textwidth}{No new pole and no lost pole under the theorem hypotheses} &
\parbox{0.25\textwidth}{Continue the physical transformed solutions through the apparent singularity} \\
\parbox{0.21\textwidth}{$D(\cdot,\omega_0)\equiv0$} &
\parbox{0.22\textwidth}{Global transformation degeneracy} &
\parbox{0.23\textwidth}{Spectral equivalence may fail} &
\parbox{0.25\textwidth}{Check the kernel of $T$ on the homogeneous solution space} \\
\parbox{0.21\textwidth}{$N_-(\omega_0)=0$ or $N_+(\omega_0)=0$} &
\parbox{0.22\textwidth}{Leading endpoint term cancels} &
\parbox{0.23\textwidth}{Unit Jost normalization may fail} &
\parbox{0.25\textwidth}{Recompute the subleading endpoint asymptotics} \\
\parbox{0.21\textwidth}{Threshold or branch point} &
\parbox{0.22\textwidth}{Meromorphic description may fail} &
\parbox{0.23\textwidth}{Laurent-pole theorem does not apply directly} &
\parbox{0.25\textwidth}{Use a separate threshold or branch analysis} \\
\parbox{0.21\textwidth}{$D$-zero reaches an endpoint or an original singular point} &
\parbox{0.22\textwidth}{Ordinary-point assumptions fail} &
\parbox{0.23\textwidth}{No conclusion from the local theorem alone} &
\parbox{0.25\textwidth}{Redo the local or endpoint asymptotic analysis} \\
\hline\hline
\end{tabular}
\end{center}
\end{table*}

Two further points are worth keeping in mind.  If $A$ or $B$ has a pole
as a function of $\omega$, then the observable itself may introduce an
operator pole that is not a pole of the original spectral problem.
Also, if new boundary conditions are imposed after the transformation
instead of transporting the original boundary lines, one has defined a
different spectral problem and the invariance theorem does not apply.

A transformation frequency described as ``algebraically special'' is
not automatically exceptional for every observable.
It becomes exceptional only if the particular transformation develops
a kernel, a normalization failure, or another degeneracy there.

This point is important in comparing a general observable with a
Chandrasekhar or Darboux transformation.


\section{Schwarzschild applications}
\label{sec:schwarzschild}

We now apply the general results to the two standard gravitational
master equations of Schwarzschild spacetime
\cite{ReggeWheeler1957,Zerilli1970a,Zerilli1970b,
Moncrief1974,Chandrasekhar1983}.

The purpose is not to recompute their quasinormal spectra.
We use them because the derivative observable gives a direct physical
example of determinant zeros.

\subsection{Regge--Wheeler equation}
\label{sec:rw}

For axial perturbations with $\ell\ge2$,
\begin{equation}
    \frac{\dd^2\Psi_{\rm RW}}{\dd r_*^2}
    +
    \left[
       \omega^2-V_{\rm RW}(r)
    \right]
    \Psi_{\rm RW}
    =0,
    \label{eq:RWeq}
\end{equation}
where
\begin{equation}
    f(r)=1-\frac{2M}{r},
\end{equation}
and
\begin{equation}
    V_{\rm RW}(r)
    =
    f(r)
    \left[
       \frac{\ell(\ell+1)}{r^2}
       -
       \frac{6M}{r^3}
    \right].
    \label{eq:RWpot}
\end{equation}

Take
\begin{equation}
    T=\partial_{r_*}.
\end{equation}
Then
\begin{equation}
    \boxed{
    D_{\rm RW}(r,\omega)
    =
    \omega^2-V_{\rm RW}(r).}
    \label{eq:DRW}
\end{equation}
Thus the determinant zeros are exactly the ordinary turning points of
the original Regge--Wheeler equation.

For real nonzero $\omega$, the potential is positive outside the
horizon and tends to zero at both ends.
Writing
\begin{equation}
    L_\ell=\ell(\ell+1),
    \qquad
    x=\frac{r}{M},
\end{equation}
we have
\begin{equation}
    M^2V_{\rm RW}
    =
    \frac{
       (x-2)(L_\ell x-6)
    }{x^4}.
\end{equation}
Differentiating gives
\begin{equation}
    \frac{\dd}{\dd x}
    \left(
      M^2V_{\rm RW}
    \right)
    =
    -\frac{2}{x^5}
    \left[
       L_\ell x^2
       -(3L_\ell+9)x
       +24
    \right].
\end{equation}
The potential has one exterior maximum.
Its location is
\begin{equation}
    x_{\rm pk}^{\rm RW}
    =
    \frac{
       3(L_\ell+3)
       +
       \sqrt{
          9(L_\ell+3)^2-96L_\ell
       }
    }{
       2L_\ell
    }.
    \label{eq:RWpeak}
\end{equation}

Therefore:

\begin{enumerate}
    \item if
    \begin{equation}
        0<\omega^2<V_{\max},
    \end{equation}
    there are two simple exterior turning points;

    \item if
    \begin{equation}
        \omega^2=V_{\max},
    \end{equation}
    they merge into one double turning point;

    \item if
    \begin{equation}
        \omega^2>V_{\max},
    \end{equation}
    there is no real exterior turning point.
\end{enumerate}

For the derivative equation, the two simple turning points have
Frobenius exponents
\begin{equation}
    0,\qquad2,
\end{equation}
while the double turning point has
\begin{equation}
    0,\qquad3.
\end{equation}

For complex quasinormal frequencies, the turning points generally move
into the complex $r$-plane.
They should not be treated as real radial turning points unless the
frequency makes $\omega^2$ real and positive.

At nonzero frequency the physical boundary solutions satisfy
\begin{equation}
    \Psi_{\rm in}
    \sim
    e^{-i\omega r_*},
    \qquad
    r_*\rightarrow-\infty,
\end{equation}
and
\begin{equation}
    \Psi_{\rm up}
    \sim
    e^{+i\omega r_*},
    \qquad
    r_*\rightarrow+\infty.
\end{equation}
Differentiation gives
\begin{equation}
    N_-=-i\omega,
    \qquad
    N_+=+i\omega.
    \label{eq:RWN}
\end{equation}

Hence, for $\omega\neq0$, the derivative map preserves the ingoing
and outgoing boundary lines.

\begin{corollary}[Regge--Wheeler derivative observable]
\label{cor:RW}
On a nonzero-frequency domain where the outgoing Regge--Wheeler
inverse is meromorphic,
\begin{equation}
    T=\partial_{r_*}
\end{equation}
has no global determinant degeneracy at finite $\omega$ because
$V_{\rm RW}(r)$ is not constant.

Every finite turning point
\begin{equation}
    \omega^2=V_{\rm RW}(r_0)
\end{equation}
is therefore an isolated determinant zero.
It produces an apparent singularity of the derivative equation but
does not create a quasinormal-mode zero and does not change the pole
order of the physical derivative response.
\end{corollary}

The frequency
\begin{equation}
    \omega=0
\end{equation}
must be excluded from this corollary.
The endpoint factors in \eqref{eq:RWN} vanish there, and the
Schwarzschild frequency-domain Green function has nontrivial
low-frequency and branch-cut structure
\cite{Leaver1986,ChingLeungSuenYoung1995,
Andersson1997,CasalsOttewill2012}.

\subsection{Zerilli equation}
\label{sec:zerilli}

For even-parity perturbations define
\begin{equation}
    n=\frac{(\ell-1)(\ell+2)}{2}.
\end{equation}
The Zerilli equation is
\begin{equation}
    \frac{\dd^2\Psi_Z}{\dd r_*^2}
    +
    \left[
      \omega^2-V_Z(r)
    \right]
    \Psi_Z
    =0,
    \label{eq:Zeq}
\end{equation}
with
\begin{align}
    V_Z
    ={}&
    \frac{2f}{r^3(nr+3M)^2}
    \Bigl[n^2(n+1)r^3+3n^2Mr^2
    \nonumber\\
    &\hspace{2.0cm}+9nM^2r+9M^3\Bigr].
    \label{eq:Zpot}
\end{align}

For $T=\partial_{r_*}$,
\begin{equation}
    \boxed{
    D_Z(r,\omega)=\omega^2-V_Z(r).}
\end{equation}
The Zerilli potential is positive in the exterior, tends to zero at
both ends, and has one exterior maximum for every $\ell\ge2$.  A short
proof of the last statement is given in Appendix~\ref{app:Zpeak}.
Therefore the real turning-point pattern is the same as in the
Regge--Wheeler case: two simple turning points below the maximum, one
double turning point at the maximum, and none above it.  The actual
turning-point radii are different because $V_Z$ and $V_{\rm RW}$ are
different potentials; this has no spectral consequence.

The endpoint factors are again
\begin{equation}
    N_-=-i\omega,
    \qquad
    N_+=+i\omega.
\end{equation}
Hence the general theorem applies without modification at nonzero
frequency.

\begin{corollary}[Zerilli derivative observable]
\label{cor:Z}
On every nonzero-frequency meromorphic Jost domain, finite zeros of
$\omega^2-V_Z(r)$ are apparent singularities of the derivative equation.
They do not change the Zerilli resonance zeros or the pole order of the
physical derivative response.  The zero-frequency limit requires a
separate threshold analysis.
\end{corollary}


\section{Darboux and Chandrasekhar transformations}
\label{sec:darboux}

A true Darboux transformation is more restrictive than the general
observable studied here
\cite{Darboux1882,Crum1955,Sukumar1985,
CooperKhareSukhatme1995}.

Take first
\begin{equation}
    B=1,
    \qquad
    T=\partial_x+F(x).
\end{equation}
Then
\begin{equation}
    D
    =
    q-F'-pF+F^2.
\end{equation}
The general source identity becomes
\begin{equation}
    \widetilde LT
    =
    \left(
       T-\frac{D'}{D}
    \right)L.
\end{equation}
Therefore
\begin{equation}
    \widetilde LT=TL
\end{equation}
requires
\begin{equation}
    D'=0.
\end{equation}

For a general nonzero $B$, the condition is
\begin{equation}
    \boxed{
    \left(
      \frac{D}{B^2}
    \right)'=0.}
    \label{eq:DarbouxCondition}
\end{equation}

This gives a simple algebraic test for when a general observable is
also an intertwiner.

Consider the Schr\"odinger family
\begin{equation}
    L(\omega)
    =
    \partial_x^2
    +
    \omega^2-V(x).
\end{equation}
Let $\phi$ be a Darboux seed at the frequency $\omega_*$,
\begin{equation}
    \phi''
    +
    \left[
       \omega_*^2-V
    \right]\phi
    =0,
\end{equation}
and take
\begin{equation}
    T
    =
    \partial_x
    -
    \frac{\phi'}{\phi}.
\end{equation}
A direct calculation gives
\begin{equation}
    \boxed{
    D
    =
    \omega^2-\omega_*^2.}
    \label{eq:Dseed}
\end{equation}

The determinant is independent of $x$.
At
\begin{equation}
    \omega^2=\omega_*^2,
\end{equation}
it vanishes everywhere and
\begin{equation}
    T\phi=0.
\end{equation}
Thus the usual deleted or exceptional Darboux state is exactly the
global-degeneracy case
\begin{equation}
    D(\cdot,\omega)\equiv0
\end{equation}
of the present theory.

The Regge--Wheeler and Zerilli equations are connected by the
Chandrasekhar transformation
\cite{Chandrasekhar1975,Chandrasekhar1983}.  Its Darboux form was
discussed in Ref.~\cite{Glampedakis2017}.  The algebraically special
frequency was studied in Refs.~\cite{Leung1999,Leung2001,Maassen2000}.
We do not repeat those results here.

The important point is simpler.
A genuine Chandrasekhar/Darboux map is a special member of the larger
class
\begin{equation}
    A+B\partial_x.
\end{equation}
Because it satisfies an intertwining condition, it can have stronger
spectral and Green-function identities than a generic observable.

This also explains why generalized first-order maps need care.
If one transforms the equation but imposes new boundary conditions
that are not the images of the old ones, QNM isospectrality can fail
\cite{Yurov2019}.
Our global theorem avoids this problem by following the physical
boundary lines under $T$.


\section{Implication for Teukolsky-type radial equations}
\label{sec:kerr}

The main motivation for this work came from differentiated radial
equations in Kerr perturbation theory.

The Teukolsky equation separates into radial and angular equations
\cite{Teukolsky1972,Teukolsky1973,
PressTeukolsky1973}.
The radial equation can be written in confluent-Heun form, and analytic
representations are also available through the Mano--Suzuki--Takasugi
method
\cite{ManoSuzukiTakasugi1996,ManoTakasugi1997}.

Aly and Stojkovic studied the equation in horizon-penetrating
coordinates and found that the equation satisfied by the first radial
derivative contains an additional regular singular point
\cite{AlyStojkovic2023}.
Their reconstructed metric quantities contain both the radial function
and its derivative.

Write a normalized radial equation as
\begin{equation}
    R''
    +
    p(r,\omega)R'
    +
    q(r,\omega)R
    =0.
    \label{eq:TeukNorm}
\end{equation}
For
\begin{equation}
    T=\partial_r,
\end{equation}
the determinant is
\begin{equation}
    \boxed{
    D=q.}
\end{equation}

Suppose
\begin{equation}
    q(r_0,\omega_0)=0
\end{equation}
at an ordinary point of the original radial equation, and suppose the
zero is isolated.
Then the equation for $R'$ has an apparent singularity.
For a simple zero its exponents are
\begin{equation}
    0,\qquad2.
\end{equation}

The global conclusion follows from the same theorem.
If
\begin{equation}
    q(r_0,\omega_0)=0,
    \qquad
    q(\cdot,\omega_0)\not\equiv0,
\end{equation}
this equality by itself is not a QNM condition and does not produce a
pole of
\begin{equation}
    \partial_rR_L(\omega).
\end{equation}

A reconstructed quantity of the form
\begin{equation}
    \mathcal H
    =
    a(r,\omega)R
    +
    b(r,\omega)R'
\end{equation}
also remains regular at the apparent derivative singularity as long as
$a$ and $b$ are regular there and no separate physical singularity is
present.

This does not yet prove a complete Kerr spectral theorem for every
Teukolsky-derived observable.
For that, one must also check:

\begin{enumerate}
    \item the transformed horizon boundary line;

    \item the transformed infinity boundary line;

    \item the chosen analytic continuation in $\omega$;

    \item threshold or superradiant endpoint degeneracies;

    \item possible transformation constants, including
    Teukolsky--Starobinsky factors
    \cite{StarobinskyChurilov1973,TeukolskyPress1974};

    \item branch points of the outgoing inverse.
\end{enumerate}

Transformations such as the Sasaki--Nakamura map are designed
specifically to improve the asymptotic radial problem
\cite{SasakiNakamura1982,Nakano2016}.
Whiting's mode-stability transformation is another example where the
transformation is adapted to the global Kerr boundary problem
\cite{Whiting1989}.

The conclusion here is limited to this question.  An isolated extra
singular point in the derivative equation is not, by itself, a new
spectral singularity.  A problem can occur only if the transformation
also fails globally, or if one of the endpoint or analytic assumptions
used above breaks down.


\section{Symbolic and numerical checks}
\label{sec:numerics}

The main results of this paper are analytic.
We use symbolic and numerical calculations only as consistency checks.

The identities
\begin{equation}
    \widetilde LT=\mathcal SL,
\end{equation}
\begin{equation}
    UT
    =
    I-\frac{B^2}{D}L,
\end{equation}
and
\begin{equation}
    TU
    =
    I-\frac{B^2}{D}\widetilde L
\end{equation}
were verified by direct symbolic expansion of all derivative
coefficients.

We also use the Regge--Wheeler equation to check the continuation
through a turning point and the Green-function cancellation.

\subsection{Continuation through a turning point}

Take
\begin{equation}
    M=1,
    \qquad
    \ell=2,
    \qquad
    \omega=0.2.
\end{equation}
The inner real turning point is
\begin{equation}
    r_0
    =
    2.120267923466549.
\end{equation}

Let
\begin{equation}
    X=\frac{\dd\Psi}{\dd r_*}.
\end{equation}
Instead of integrating the singular scalar equation for $X$ directly
through the turning point, we use the regular first-order system
\begin{equation}
    \frac{\dd\Psi}{\dd r}
    =
    \frac{X}{f},
\end{equation}
\begin{equation}
    \frac{\dd X}{\dd r}
    =
    -\frac{Q\Psi}{f}.
\end{equation}

With the normalization
\begin{equation}
    \Psi(3)=1,
    \qquad
    X(3)=0.2,
\end{equation}
the regular system gives
\begin{equation}
    \Psi(r_0)
    =
    0.8329475239552365,
\end{equation}
and
\begin{equation}
    X(r_0)
    =
    -0.0432096167335503.
\end{equation}

The local series for the derivative begins as
\begin{align}
    X(t)
    ={}&
    -0.0432096167
    +2.1500600225\,t^2
    \nonumber\\
    &
    -14.06526669\,t^3
    +92.2495386\,t^4
    +\cdots.
\end{align}
There is no linear term, in agreement with
\begin{equation}
    X'(r_0)=0
\end{equation}
for a simple turning point.

Using DOP853 with relative tolerance
$2\times10^{-12}$ and absolute tolerance
$2\times10^{-14}$ gives the convergence shown in
Table~\ref{tab:localconv}.

\begin{table}[ht]
\caption{
Comparison of the local Frobenius series with direct integration across
the inner Regge--Wheeler turning point.
}
\label{tab:localconv}
\begin{ruledtabular}
\begin{tabular}{ccc}
series order
&
maximum error, right
&
maximum error, left
\\
4
&
$2.12\times10^{-4}$
&
$2.89\times10^{-5}$
\\
6
&
$1.44\times10^{-6}$
&
$1.95\times10^{-7}$
\\
8
&
$9.58\times10^{-9}$
&
$1.30\times10^{-9}$
\\
10
&
$6.37\times10^{-11}$
&
$8.66\times10^{-12}$
\end{tabular}
\end{ruledtabular}
\end{table}

The calculation confirms that the physical derivative passes smoothly
through the singular point of its scalar equation.

\subsection{Green response}

We next test the global identity.
We use the tortoise-coordinate interval
\begin{equation}
    r_*\in[-25,70]
\end{equation}
with the same values
\begin{equation}
    M=1,
    \qquad
    \ell=2,
    \qquad
    \omega=0.2.
\end{equation}
The two real turning points are
\begin{equation}
    r_{*1}
    =
    -3.5020930964663,
\end{equation}
and
\begin{equation}
    r_{*2}
    =
    13.366435440180.
\end{equation}

The left homogeneous solution is initialized by
\begin{equation}
    u_-(-25)
    =
    e^{-i\omega(-25)},
\end{equation}
\begin{equation}
    u_-'(-25)
    =
    -i\omega u_-(-25),
\end{equation}
while the right solution is initialized by
\begin{equation}
    u_+(70)
    =
    e^{+i\omega70},
\end{equation}
\begin{equation}
    u_+'(70)
    =
    +i\omega u_+(70).
\end{equation}

Both solutions are integrated with DOP853 using
relative tolerance $10^{-11}$, absolute tolerance $10^{-13}$,
and maximum step $0.2$.

We choose the smooth compact source
\begin{equation}
 f_s(x)=
 \begin{cases}
 \displaystyle
 \exp
 \left[
  -\frac{1}{
     1-((x-5)/2)^2
  }
 \right],
 & |x-5|<2,
 \\[2mm]
 0,
 & \text{otherwise}.
 \end{cases}
 \label{eq:bumpsource}
\end{equation}

The support does not meet either turning point.
We compute the original outgoing response
\begin{equation}
    \Psi=R_Lf_s
\end{equation}
and differentiate it directly.

Independently, for
\begin{equation}
    T=\partial_x,
\end{equation}
we use
\begin{equation}
    \mathcal Sf_s
    =
    f_s'
    -
    \frac{Q'}{Q}f_s
\end{equation}
inside the transformed Green representation.

The two answers converge to one another.
With composite trapezoidal quadrature the relative discrepancy
decreases approximately by a factor of four when the mesh size is
halved, as expected for a second-order quadrature rule.

\begin{table}[ht]
\caption{
Convergence of the transformed Green response to the direct derivative
of the original response.
}
\label{tab:greenconv}
\begin{ruledtabular}
\begin{tabular}{cc}
grid points
&
maximum relative discrepancy
\\
3001
&
$2.15\times10^{-4}$
\\
6001
&
$5.39\times10^{-5}$
\\
12001
&
$1.35\times10^{-5}$
\\
24001
&
$3.37\times10^{-6}$
\\
48001
&
$8.42\times10^{-7}$
\end{tabular}
\end{ruledtabular}
\end{table}

The physical derivative response remains smooth at both turning points.

Finally, we examine the singular unit-source transformed kernel near
the second turning point.
For fixed $x=0<r_{*2}$,
\begin{equation}
    G_{\widetilde L}^{(\delta)}(x,s)
    =
    \frac{
      u_-'(x)u_+'(s)
    }{
      Q(s)W
    }.
\end{equation}
As $s\to r_{*2}$ this kernel grows like $1/Q$.
The two separate terms in
\begin{equation}
    \left(
      -\partial_s-\frac{Q'}{Q}
    \right)
    G_{\widetilde L}^{(\delta)}
\end{equation}
are even larger, but they cancel.

\begin{table*}[t]
\caption{
Cancellation of the artificial transformed Green singularity near the
second Regge--Wheeler turning point.
}
\label{tab:kernelcancel}
\begin{ruledtabular}
\begin{tabular}{ccccc}
$\epsilon=s-r_{*2}$
&
$|G_{\widetilde L}^{(\delta)}|$
&
$|-\partial_sG_{\widetilde L}^{(\delta)}|$
&
$|-(Q'/Q)G_{\widetilde L}^{(\delta)}|$
&
combined value
\\
$10^{-1}$
&
$7.79$
&
$7.73\times10^{1}$
&
$7.73\times10^{1}$
&
$3.344\times10^{-2}$
\\
$10^{-2}$
&
$7.73\times10^{1}$
&
$7.73\times10^{3}$
&
$7.73\times10^{3}$
&
$3.359\times10^{-2}$
\\
$10^{-3}$
&
$7.73\times10^{2}$
&
$7.73\times10^{5}$
&
$7.73\times10^{5}$
&
$3.360\times10^{-2}$
\\
$10^{-4}$
&
$7.73\times10^{3}$
&
$7.73\times10^{7}$
&
$7.73\times10^{7}$
&
$3.360\times10^{-2}$
\end{tabular}
\end{ruledtabular}
\end{table*}

The numerical result is not part of the proof.
It only shows directly what Theorem~\ref{thm:GreenExact} predicts:
the unit-source scalar kernel is singular, while the physical composed
kernel is regular.


\section{Discussion and conclusion}
\label{sec:discussion}

We have studied first-order differential observables of
parameter-dependent second-order wave equations.  The main point is
that a zero of the transformation determinant at one spatial point is
not a spectral condition.  Under the assumptions used in this paper,
such a zero can create an apparent singularity of the transformed
scalar equation, but it does not create or remove a pole of the
physical transformed response.

A genuine failure is stronger.  It can occur when the transformation
loses a homogeneous solution globally, when the transported boundary
line fails, or when the frequency reaches a threshold, branch point, or
another exceptional parameter excluded from the theorem.  On a
connected regular domain, the loss of a homogeneous solution is
equivalent to $D(\cdot,\omega)$ vanishing identically.

The Green-function calculation makes the same distinction from the
inhomogeneous side.  The unit-source Green function of the transformed
scalar equation can contain a $1/D$ singularity.  This does not mean
that the physical observable is singular.  The source is transformed
at the same time as the field, and the adjoint Green identity proved in
Sec.~\ref{sec:resolvent} shows the cancellation explicitly.

Several ingredients used here are classical.  Apparent singularities
from scalar reductions, derivative Heun equations, Evans and Jost
functions, and Green functions of Darboux partners are all established
subjects
\cite{DubrovinMazzocco2007,vanderPutSaito2009,Filipuk2020,
AlexanderGardnerJones1990,GesztesyLatushkinMakarov2007,
SamsonovSukumarPupasov2005,SamsonovPupasov2005}.
The new part of this paper is to connect these ideas for a general
first-order observable that does not have to be an intertwiner.  We
prove a global kernel criterion, a spectral result for the transformed
boundary solutions, preservation of resolvent pole order, and an exact
Green-function cancellation identity.

For the Regge--Wheeler and Zerilli derivative observables, finite
turning points are therefore apparent singularities rather than
nonzero-frequency QNM conditions.  The point $\omega=0$ remains
separate because the endpoint normalization degenerates and the
Schwarzschild Green function has threshold and branch structure.  The
same local conclusion applies to an isolated zero in a differentiated
Teukolsky radial equation, but a full Kerr spectral statement would
also require a separate check of the horizon and infinity boundary
maps and any special transformation constants.

The framework is not limited to black holes.  It applies to other
second-order scattering problems whenever the physical quantity of
interest is a first-order differential observable and the outgoing
inverse has a meromorphic realization.  The practical lesson is simple:
a singularity in the scalar equation for an observable should not be
identified with a physical spectral singularity until the global
boundary problem and the transformed source have also been checked.


\appendix

\section{Expanded coefficient identities}
\label{app:coefficients}

For reference,
\begin{equation}
    C=A'-Bq,
\end{equation}
\begin{equation}
    E=A+B'-Bp,
\end{equation}
\begin{equation}
    F=C'-Eq,
\end{equation}
and
\begin{equation}
    G=C+E'-Ep.
\end{equation}

Expanding $F$ gives
\begin{equation}
    F
    =
    A''
    -
    B'q
    -
    Bq'
    -
    Aq
    -
    B'q
    +
    Bpq.
\end{equation}
Similarly,
\begin{equation}
    G
    =
    (A'-Bq)
    +
    (A'+B''-B'p-Bp')
    -
    (A+B'-Bp)p.
\end{equation}

The compact definitions are more useful in the main text because
\begin{equation}
    \widetilde q
    =
    \frac{GC-FE}{D}
\end{equation}
then remains readable.

A direct differentiation of
\begin{equation}
    D=AE-BC
\end{equation}
also gives
\begin{equation}
    FB-GA
    =
    pD-D'.
\end{equation}
This is the algebraic form of the Abel-identity check used in
Eq.~\eqref{eq:ptilde}.

The first B\'ezout identity follows directly from
\begin{align}
    (E-B\partial_x)(A+B\partial_x)
    &=
    D
    -
    B^2
    \left(
       \partial_x^2
       +p\partial_x
       +q
    \right)
    \nonumber\\
    &=
    D-B^2L.
\end{align}

\section{First-jet constraint at a simple determinant zero}
\label{app:firstjet}

Assume that $A$ and $B$ do not have a common zero at $x_0$ and that
$D$ has a simple zero there.
Then
\begin{equation}
    \operatorname{rank}M(x_0)=1
\end{equation}
unless the complete matrix $M(x_0)$ vanishes.

Every physical transformed solution therefore satisfies one linear
relation between
\begin{equation}
    z(x_0)
\end{equation}
and
\begin{equation}
    z'(x_0).
\end{equation}
The second independent solution has not disappeared.
Its missing datum appears at the next derivative order.

For the derivative transformation,
\begin{equation}
    M
    =
    \begin{pmatrix}
       0&1\\
       -q&-p
    \end{pmatrix}.
\end{equation}
At a simple zero of $q$,
\begin{equation}
    z'(x_0)
    =
    -p(x_0)z(x_0).
\end{equation}
For a Schr\"odinger equation, $p=0$, so
\begin{equation}
    X'(x_0)=0.
\end{equation}
The next derivative then reconstructs the missing value of the original
master function through Theorem~\ref{thm:turning}.

\section{The Zerilli potential has one exterior maximum}
\label{app:Zpeak}

We give the algebra used in Sec.~\ref{sec:zerilli}.
Set
\begin{equation}
    x=\frac{r}{M},
    \qquad
    y=x-2,
\end{equation}
and
\begin{equation}
    n
    =
    \frac{(\ell-1)(\ell+2)}{2}.
\end{equation}
For $\ell\ge2$ we have
\begin{equation}
    n\ge2.
\end{equation}

After removing a positive rational factor, the sign of
$\dd V_Z/\dd y$ is the opposite of the polynomial
\begin{align}
 P_n(y)
 ={}&
 2n^3(n+1)y^5
 +
 n^3(14n+23)y^4
 \nonumber\\
 &+
 n^2(32n^2+74n+39)y^3
 \nonumber\\
 &+
 n(16n^3+52n^2+108n+99)y^2
 \nonumber\\
 &+
 (-32n^4-104n^3-36n^2+126n+81)y
 \nonumber\\
 &-
 (32n^4+128n^3+192n^2+144n+54).
 \label{eq:Pn}
\end{align}

For $n\ge2$, the first four coefficients are positive.
The coefficient of $y$ is negative because
\begin{equation}
    32n^4
    +
    104n^3
    +
    36n^2
    -
    126n
    -
    81
    >0.
\end{equation}
The constant coefficient is also negative.

The signs in descending order are therefore
\begin{equation}
    +,+,+,+,-,-.
\end{equation}
There is exactly one sign change.
Descartes' rule of signs implies that $P_n$ has exactly one positive
root.

Since
\begin{equation}
    P_n(0)<0
\end{equation}
and
\begin{equation}
    P_n(y)>0
\end{equation}
for sufficiently large $y$, the Zerilli potential first increases and
then decreases.
It therefore has exactly one exterior maximum.

\section{Distributional form of the transformed source}
\label{app:distribution}

Write
\begin{equation}
    \mathcal S
    =
    B\partial_x+H,
\end{equation}
where
\begin{equation}
    H
    =
    A+2B'
    -
    B\frac{D'}{D}.
\end{equation}

For a delta source at $s$,
\begin{align}
    \mathcal S_x\delta(x-s)
    ={}&
    B(x)\delta'(x-s)
    +H(x)\delta(x-s)
    \nonumber\\
    ={}&
    B(s)\delta'(x-s)
    +\left[H(s)-B'(s)\right]\delta(x-s)
    \nonumber\\
    ={}&
    B(s)\delta'(x-s)
    \nonumber\\
    &+
    \left[
       A+B'-B\frac{D'}{D}
    \right]_{s}
    \delta(x-s).
    \label{eq:distributionalsource}
\end{align}

The coefficient multiplying $\delta$ is exactly the zeroth-order
coefficient of the formal adjoint
\begin{equation}
    \mathcal S^\dagger
    =
    -B\partial_x
    +
    A+B'
    -
    B\frac{D'}{D}.
\end{equation}

This is the distributional form used in the Green-kernel calculation.

\section*{Conflict of interest}

The author declares no conflicts of interest.

\section*{Ethics statement}

This work did not involve human participants, animals, personal data, or other research requiring ethical approval.

\section*{Funding}

The author received no specific funding for this work

\section*{Data availability}

No external data sets were used in this work.  All numerical values
reported in Sec.~\ref{sec:numerics} were generated from the equations,
parameters, source function, integration interval, and solver tolerances
stated in the text.  No separate research data set was created.

\end{document}